\documentclass[oneside,english]{amsart}
\usepackage[T1]{fontenc}
\usepackage[latin9]{inputenc}
\usepackage{color}
\usepackage{amsthm}
\usepackage{amssymb}
\usepackage{graphicx}

\makeatletter

\providecommand{\tabularnewline}{\\}

\numberwithin{equation}{section}
\numberwithin{figure}{section}
\theoremstyle{plain}
\newtheorem{thm}{\protect\theoremname}
  \theoremstyle{plain}
  \newtheorem{lem}[thm]{\protect\lemmaname}
  \theoremstyle{definition}
  \newtheorem{defn}[thm]{\protect\definitionname}

\usepackage{babel}
\providecommand{\lemmaname}{Lemma}
\providecommand{\theoremname}{Theorem}

\makeatother

\usepackage{babel}
  \providecommand{\definitionname}{Definition}
  \providecommand{\lemmaname}{Lemma}
\providecommand{\theoremname}{Theorem}

\begin{document}

\title{Two-Walks degree assortativity in Graphs and Networks}

\author{Alfonso Allen-Perkins, Juan Manuel Pastor, Ernesto Estrada}

\begin{abstract}

Degree ssortativity is the tendency for nodes of high degree (resp. low degree) in a graph to be connected to high degree nodes (resp. to low degree ones). It is usually quantified by the Pearson correlation coefficient of the degree-degree correlation. Here we extend this concept to account for the effect of second neighbours to a given node in a graph. That is, we consider the two-walks degree of a node as the sum of all the degrees of its adjacent nodes. The two-walks degree assortativity of a graph is then the Pearson correlation coefficient of the two-walks degree-degree correlation. We found here analytical expression for this two-walks degree assortativity index as a function of contributing subgraphs. We then study all the 261,000 connected graphs with 9 nodes and observe the existence of assortative-assortative and disassortative-disassortative graphs according to degree and two-walks degree, respectively. More surprinsingly, we observe a class of graphs which are degree disassortative and two-walks degree assortative. 
We explain the existence of some of these graphs due to the presence of certain topological features, such as a node of low-degree connected to high-degree ones. More importantly, we study a series of 49 real-world networks, where we observe the existence of the disassortative-assortative class in several of them. In particular, all biological networks studied here were in this class. We also conclude that no graphs/networks are possible with assortative-disassortative structure.

\end{abstract}

\maketitle

\section{Introduction}

Networks represent the topological skeleton of a wide range of systems
in nature and society \cite{Albert02,Newman03a,Boccaletti06,Strogatz01}.
The characterization of their structure is crucial since it shapes
the evolutionary, functional, and dynamical processes that take place
in those systems \cite{Strogatz01,Barrat08,Costa07}.

It is well known that links generally do not connect nodes regardless
of their characteristics. In social networks, for instance, evidence
suggests that individuals prefer to associate with others of similar
age, religion, education or occupation as themselves \cite{McPherson01}.
Assortativity or assortative mixing is a graph metric that refers
to the tendency for nodes in networks to be connected to other nodes
that are similar (or different) to themselves in some way \cite{Newman03b}.
Typically, it is determined for the \textit{degree} (i.e. the number
of direct neighbours, \textit{k}) of the nodes in the network \cite{Newman02,Pastor01,Estrada11,Piraveenan08}.
The tendency for high-degree nodes to associate preferentially with
other high-degree nodes plays a major role in many important processes,
such as epidemic spreading, synchronization or network robustness,
among others \cite{Newman02,Newman03c,Eguiluz02,Boguna03,Dibernardo06}.
However, assortativity may be applied to any characteristics of a
node, including non-topological vertex properties, such as language
or race \cite{Newman03b}. Most of the research done in this area
has been summarized in the review of Noldus \textit{et al.} \cite{Noldus15}. Other extensions to account for interactions beyond the nearest-neighbours have also been proposed in the recent literature \cite{Allen16}.

The aim of this work is to define an assortativity index that captures the influence of first and second neighbours of a node. We then express this two-walks assortativity in terms of the subgraphs contributing to it.

The paper is organized as follows. In Section \ref{Sec:Prel}, the
preliminaries are presented. In Section \ref{Sec:Neighbourhood}, the concept of two-walks degree assortativity is introduced and analysed. Main result is demostrated in Section
\ref{Sec:Main}. Numerical results are presented in Section \ref{Sec:Num}. Conclusions
are summarized in Section \ref{Sec:Conclusion}.

\section{Preliminaries}

\label{Sec:Prel}

Here we consider simple, undirected graphs $G=\left(V,E\right)$,
i.e., graphs without multiple edges, self-loops, directions or weights
in their edges. The notation used is standard and the reader can check
for instance \cite{Estrada12}. Let us define some of the measures used
in this work in order to make it self-contained. First, we define
the degree assortativity index \cite{Newman03b}. Mathematically, it is written
as: 
\begin{equation}
r_{k}=\dfrac{\frac{1}{m}\sum_{\left(i,j\right)\in E}k_{i}k_{j}-\left\{ \frac{1}{m}\sum_{\left(i,j\right)\in E}\frac{1}{2}\left[k_{i}+k_{j}\right]\right\} ^{2}}{\frac{1}{m}\sum_{\left(i,j\right)\in E}\frac{1}{2}\left[k_{i}^{2}+k_{j}^{2}\right]-\left\{ \frac{1}{m}\sum_{\left(i,j\right)\in E}\frac{1}{2}\left[k_{i}+k_{j}\right]\right\} ^{2}}\label{eq:r_k}
\end{equation}
where $k_{i}$ and $k_{j}$ are the degrees at both ends of
link $\left(i,j\right)\in E$ and $m$ is the number of links. A positive assortativity
index $r_{k}>0$ indicates the tendency of higher degree nodes in
the graph to be connected to other higher degree nodes. On the other
hand, $r_{k}<0$ indicates the tendency of higher degree nodes to
be connected to lower degree nodes. It was previously proved the following
result \cite{Estrada11}.

\begin{lem}
Let $G=\left(V,E\right)$ be a simple graph. Let $k_{i}$ be the degree
of the vertex $i$. Let $|P_{1}|$ the number of edges, $|P_{2}|$
and $|P_{3}|$ the paths of length two and three, respectively, $|C_{3}|$
be the number of triangles in $G$. Then, the assortativity coefficient
can be written combinatorially as:

\begin{equation}
r_{k}=\frac{|P_{3}|+3|C_{3}|-\frac{|P_{2}|^{2}}{|P_{1}|}}{|P_{2}|+3|S_{1,3}|-\frac{|P_{2}|^{2}}{|P_{1}|}}\label{eq:r_k_P}
\end{equation}

Let $|P_{r/s}|$ the ratio $|P_{r}|/|P_{s}|$, $|S_{1,3}|$ the number
of star graphs of four nodes, and $C=3|C_{3}|/|P_{2}|$. Then: 
\begin{enumerate}
\item assortative ($r>0$): if and only if $|P_{3/2}|+C>|P_{2/1}|$, 
\item neutral ($r=0$): if and only if $|P_{3/2}|+C=|P_{2/1}|$, and $3|S_{1,3}|-|P_{2}|(|P_{2/1}|-1)\neq0$,
and 
\item disassortative ($r<0$): if and only if $|P_{3/2}|+C<|P_{2/1}|$ 
\end{enumerate}
\end{lem}

It is worth mentioning that the denominator of Eq.~\ref{eq:r_k_P}
is non-negative. Consequently, the sign of $r_{k}$ depends only upon
the sign of the numerator, which is determined by the following structural
factors: the global clustering coefficient (i.e. $C=3|C_{3}|/|P_{2}|$),
the intermodular connectivity (i.e. $|P_{3/2}|=|P_{3}|/|P_{2}|$)
and the branching (i.e. $|P_{2/1}|=|P_{2}|/|P_{1}|$) \cite{Estrada11}.

The number of subgraphs contributing to the degree assortativity can
be obtained using the following results \cite{alon1997finding}. 

\begin{lem}
Let $G=\left(V,E\right)$ be a simple graph. Let $k_{i}$ be the degree
of the vertex $i$. Let $|C_{3}|$ be the number of triangles in $G$.
Then, the number of edges $|P_{1}|$, path of length two $|P_{2}|$
and three $|P_{3}|$ are given, respectively by

\[
|P_{1}|=\frac{1}{2}\sum_{i=1}^{n}k_{i},
\]

\[
|P_{2}|=\frac{1}{2}\sum_{i=1}^{n}k_{i}\left(k_{i}-1\right),
\]

\[
|P_{3}|=\sum_{\left(i,j\right)\in E}\left(k_{i}-1\right)\left(k_{j}-1\right)-3|C_{3}|.
\]
\end{lem}
\medskip{}

\begin{lem}
Let $G=\left(V,E\right)$ be a simple graph. Let $k_{i}$ be the degree
of the vertex $i$ in $G$. Let $A$ be the adjacency matrix of $G$. Let $|P_{1}|$ and $|P_{2}|$ be respectively the number of edges and the number of paths of length two in $G$. Let $|S_{T1D}|$ be the number of subgraphs $S_{T1D}$ in $G$ (see Table \ref{tab:subgraphs}). Let $|C_{i}|$ be the number of cycles of $i$ nodes in $G$.
Then, $|C_{3}|$, $|C_{4}|$ and $|C_{5}|$ are given, respectively by 
\end{lem}
\[
|C_{3}|=\frac{1}{6}tr\left(A^{3}\right),
\]

\[
|C_{4}|=\frac{1}{8}tr\left(A^{4}\right)-2|P_{1}|-4|P_{2}|,
\]

\[
|C_{5}|=\frac{1}{10}tr\left(A^{5}\right)-3|C_{3}|-|S_{T1S}|.
\]

\begin{lem}
Let $G=\left(V,E\right)$ be a simple graph. Let $k_{i}$ be the degree
of the vertex $i$. Let $A$ be the adjacency matrix of $G$. Let $|S_{1,3}|$, $|S_{T1S}|$, $|S_{T2S}|$, $|S_{T1D}|$, $|S_{C/}|$ and $|S_{C1S}|$ be the number of subgraphs $S_{1,3}$, $S_{T1S}$, $S_{T2S}$, $S_{T1D}$, $S_{C/}$ and  $S_{C1S}$, respectively, in $G$ (see Table \ref{tab:subgraphs}). Then, 
\end{lem}

\[
|S_{1,3}|=\frac{1}{6}\sum_{i=1}^{n}k_{i}\left(k_{i}-1\right)\left(k_{i}-2\right)
\]

\[
|S_{T1S}|=\frac{1}{2}\left(\vec{k}-\vec{2}\right)^{T}diag\left(A^{3}\right)
\]

\[
|S_{T2S}|=\frac{1}{4}\left(\left(\vec{k}-\vec{2}\right)\circ\left(\vec{k}-\vec{3}\right)\right)^{T}diag\left(A^{3}\right)
\]

\[
|S_{T1D}|=\frac{1}{2}\vec{1}^{\,T}\left(A^{2}-A^{2(D)}\right)diag\left(A^{3}\right)-6|C_{3}|-2|S_{T1S}|-4|S_{C/}|
\]
 
\[
|S_{C/}|=\vec{1}^{\,T}\left(Q\circ\left(Q-A\right)\right)\left(\vec{k}-\vec{2}\right)
\]

\[
|S_{C1S}|=\vec{1}^{\,T}\left(P-P^{(D)}\right)\left(\vec{k}-\vec{2}\right)-2|S_{C/}|,
\]

\noindent where $Q=A^{2}\circ A$, $P=\frac{1}{2}\left(A^{2}\circ(A^{2}-1)\right)$, $\vec{x}$
is an all-x vector and $\circ$ denotes the Hadamard product.


\section{Two-Walks Degree assortativity in Graphs and Networks}
\label{Sec:Neighbourhood}
Let us start by the definition of the degree of a node $i$, $k_{i}$.
The intuition behind this index is very simple. Every nearest neighbour
of the node $i$ receives an identical weight of $1$. Then, we sum
the weights of every node adjacent to $i$ to obtain $k_{i}$. Mathematically,
this corresponds to obtaining the following vector $\vec{k}$ after
assigning the unit weights to every node:

\begin{equation}
\vec{k}=A\vec{1},\label{eq:degree vector}
\end{equation}
where $\vec{1}$ is an all-ones vector. The intuition behind this
index is very simple.

It is customary to consider that not all the neighbours of one particular
node are equally important. This is the basis for instance of Katz
centrality index \cite{katz1953new}, eigenvector centrality \cite{bonacich1972factoring}, 
PageRank \cite{grin1998anatomy}, subgraph centrality \cite{estrada2005subgraph}
and so for. Then, we can consider that every neighbour of the node
$i$ is weighted according to its ``importance''. Of course, the
definition of that importance will define the way in which we will
proceed. In order to consider the current development as an extension
of the concept of node degree we simply weight every node by its own degree.
That is, now we consider the vector $\vec{k}$ as the weighting vector
for the nodes of the graph. Consequently, an extension of the concept
of degree is given by applying a similar procedure as in (\ref{eq:degree vector})
to $\vec{k}$,

\begin{equation}
\vec{k}_{2}=A\vec{k}.
\end{equation}

It is straightforward to realize that $\vec{k}_{2}=A(A\vec{1})=A^2\vec{1}$. Then, obviously, the entries of this new vector represent a new kind of centrality of the nodes which counts the number of two-walks starting at the corresponding node. Consequently, we suggest the name of "two-walks" degree for the entries of $\vec{k}_{2}$. Let us call $\tilde{k}_{i}$ the $i$th
entry of $\vec{k}_{2}$ in a graph. Notice that $\tilde{k}_{i}$ accounts for the degree of the node $i$, i.e., closed walks of length two, as well as for the number of second neighbours of this node. Then,

\begin{equation}
\tilde{k}_{i}=\sum_{j\in N\left(i\right)}k_{j},
\end{equation}
where $N\left(i\right)$ is the neighbourhood of the node $i$, i.e.,
$N\left(i\right)=\left\{ j\left|\left(j,i\right)\in E\right.\right\} $.
That is, the two-walks degree $\tilde{k}_{i}$ represents the
number of weighted neighbours that the node $i$ has, where the weight
of the nodes is given by its own degree.

Let us now define a quantity analogous to the degree assortativity
index based on the two-walks degrees instead of on the node
degrees. 
\begin{defn}
Let $G$ be a connected simple graph with adjacency matrix $A$ and
let $\tilde{k}_{i}$ be the two-walks degree of the vertex $i$.
The two-walks degree assortativity index of a graph is defined as

\begin{equation}
r_{\tilde{k}}=\dfrac{\frac{1}{m}\sum_{\left(i,j\right)\in E}\tilde{k}_{i}\tilde{k}_{j}-\left\{ \frac{1}{m}\sum_{\left(i,j\right)\in E}\frac{1}{2}\left[\tilde{k}_{i}+\tilde{k}_{j}\right]\right\} ^{2}}{\frac{1}{m}\sum_{\left(i,j\right)\in E}\frac{1}{2}\left[\tilde{k}_{i}^{2}+\tilde{k}_{j}^{2}\right]-\left\{ \frac{1}{m}\sum_{\left(i,j\right)\in E}\frac{1}{2}\left[\tilde{k}_{i}+\tilde{k}_{j}\right]\right\} ^{2}}.
\end{equation}
\end{defn}
Obviously, this quantity tell us how well connected the most important
nodes in a graph are. That is, if $r_{\tilde{k}}>0$, the graph is
two-walks degree assortative, which means that the most weighted nodes
in terms of the degree of their neighbours tend to be connected to
each other. On the other hand, if if $r_{\tilde{k}}<0$, the graph
is two-walks degree disassortative, which means that the most weighted
nodes in terms of the degree of their neighbours tend to be connected
to those with least weight. If $r_{\tilde{k}}=0$, neither of these
two tendencies is observed and we shall call such graphs neutral.

In Fig. \ref{Graphic example} we represent a graph which is strongly
disassortative for the degree ($r_{k}\approx-0.822$) but it is assortative
for the two-walks degree index ($r_{\tilde{k}}\approx0.212$). We plot
the graph with the nodes weighted by the difference between the degree
(resp. two-walks degree) minus the average degree (resp. average
two-walks degree). The negative values are colored in red and the
positive contributions in blue. The size of the nodes is proportional
to the absolute value of this difference. As can be seen in this picture
the degree-degree interaction between the nodes (left panel) is dominated
by red-blue interactions, which indicates a large number of interactions
between high degree nodes (blue ones) with low degree ones (red nodes).
This of course results in a negative degree assortativity coefficient.
On the other hand, for the two-walks degree plot the graph is
dominated by blue-blue and red-red interactions. That is, nodes of
high two-walks degree interact with each other, and low two-walks degree
nodes also interact preferentially among them. This effects result
in a two-walks degree assortativity coefficient.

\begin{center}
\begin{figure}
\begin{centering}
\includegraphics[width=0.45\textwidth]{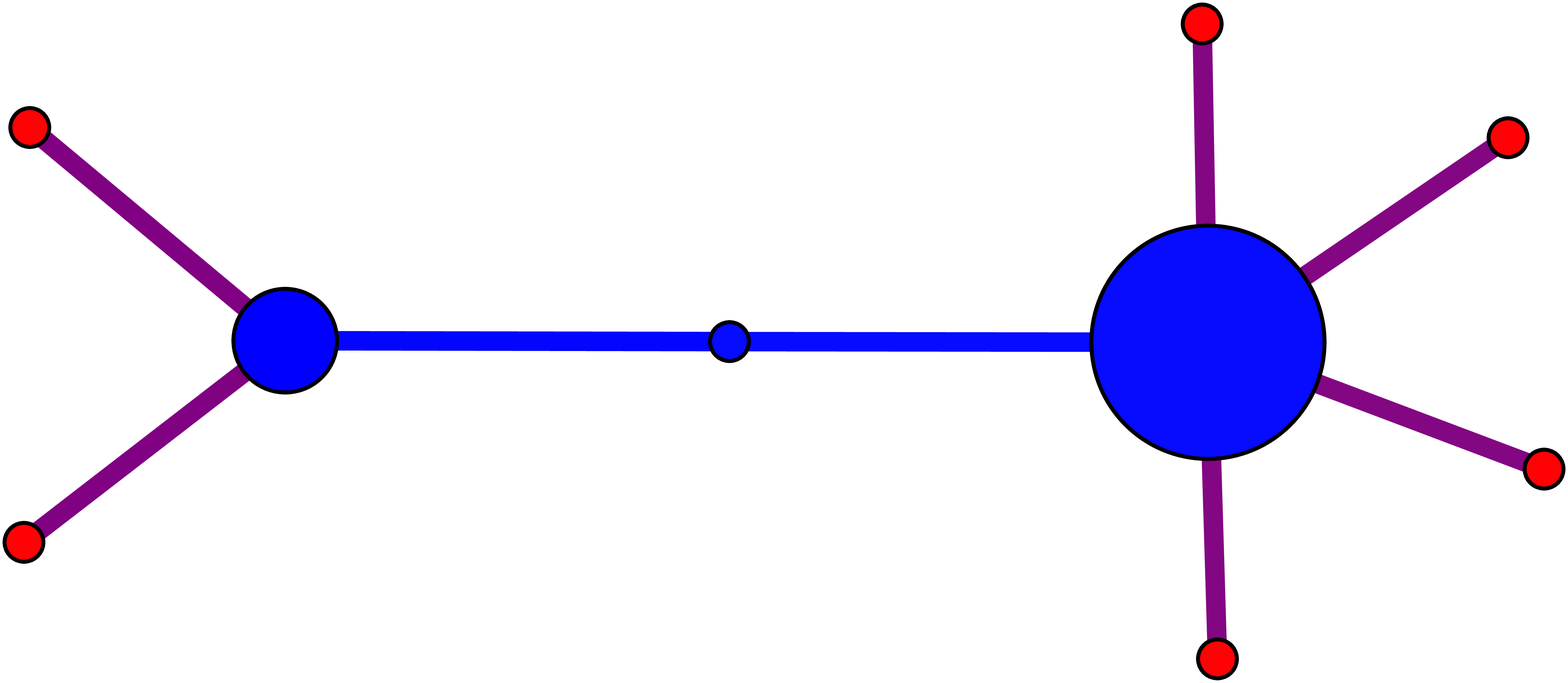}\qquad{}\includegraphics[width=0.45\textwidth]{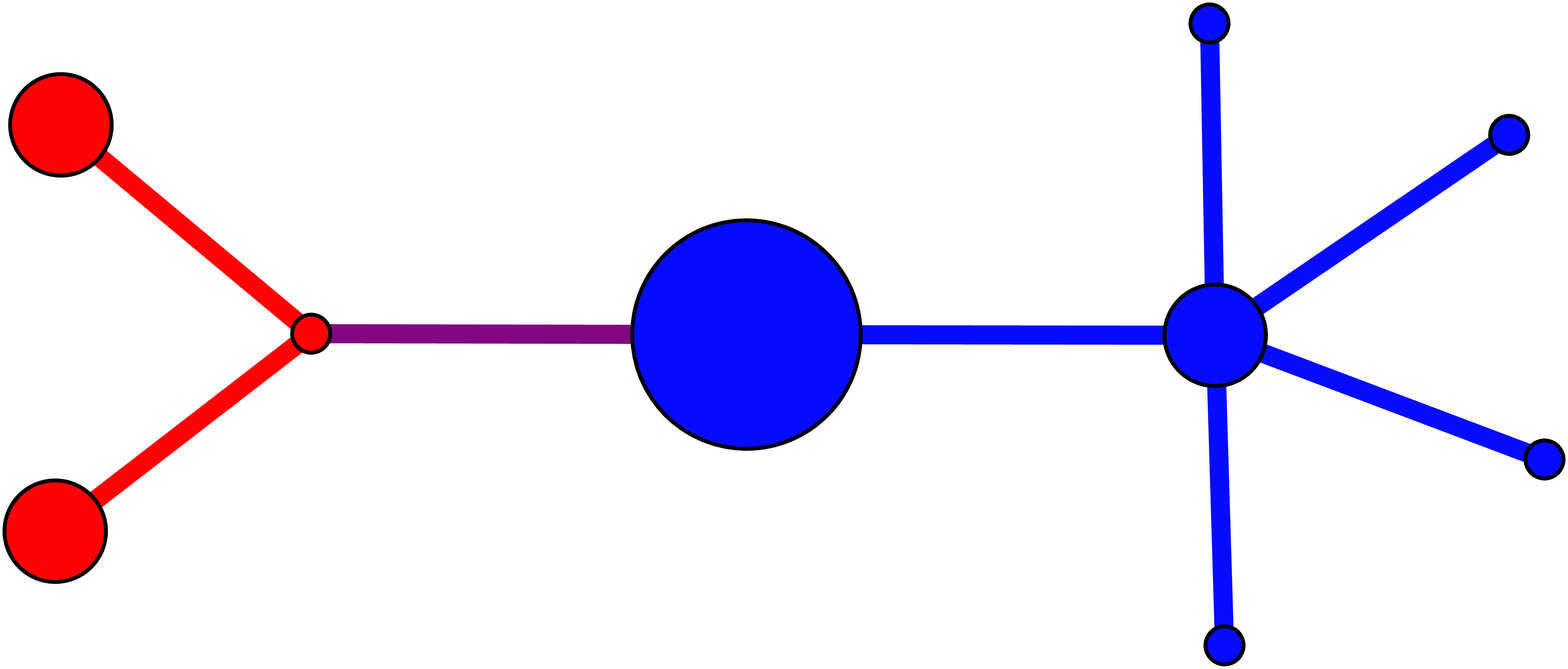} 
\par\end{centering}
\caption{Example of the structural effect that may produce a change from degree disassortative to two-walks degree assortative in a simple graph. Here the nodes are drawn in red if their degree (resp.two-walks degree) is smaller than the average degree, or blue otherwise. The size of the node is proportional to the magnitude of this difference.}
\label{Graphic example} 
\end{figure}
\par\end{center}

With the new correlation coefficient introduced here we assess the
tendency of neighbourhoods with many interactions to be connected
to other ''high-connected'' neighbourhoods.
However, in order for a graph to display a transition from degree diassortative
to two-walks degree assortative it is necessary that there are \textit{\textcolor{black}{separator
nodes}} between the high-degree nodes. The graph in Fig. \ref{Graphic example}
has a separator, which is the node of degree 2 connecting both nodes
of degree 3 and 5. A separator must be a low-degree node which connects
two or more high-degree ones. Notice that if the number of high-degree
nodes connected to the separator is too high, it will produce an increase
in its own degree, which decreases its chances of being a proper separator.
This characteristic\textendash a separator connected to two high-degree
nodes\textendash introduces disassortativity to the graph. However,
in term of the second-order correlation a separator allow the two-steps
interactions between hubs, which results in two-walks degree assortativity. 
Mathematically, it is not difficult to see that the two-walks degree is related to walks of length two between node.



It is easy to realize that the two-walks degree assortativity can be
written in matrix-vector form in the following way:

\begin{equation}
r_{\tilde{k}}=\frac{\vec{1}^{T}A^{5}\vec{1}-\left(\vec{1}^{T}A\vec{1}\right)^{-1}\left(\vec{1}^{T}A^{3}\vec{1}\right)^{2}}{\left(\vec{k}_{2}\circ\vec{k}_{2}\right)^{T}A\vec{1}-\left(\vec{1}^{T}A\vec{1}\right)^{-1}\left(\vec{1}^{T}A^{3}\vec{1}\right)^{2}}\label{eq:r_R_matrix}
\end{equation}


\begin{table}[htp]
\begin{tabular}{ccc}
\begin{minipage}[c]{4cm}%
\center{\includegraphics[width=2.5cm]{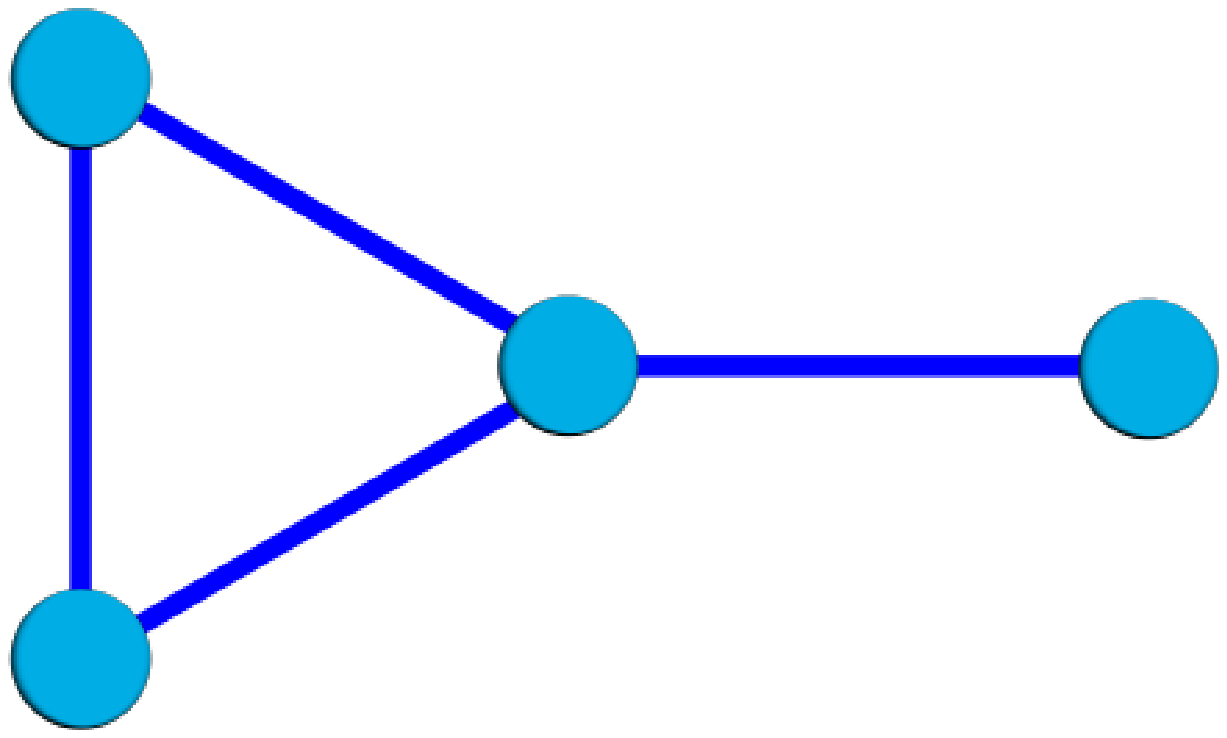}} %
\end{minipage} & %
\begin{minipage}[c]{4cm}%
\center{\includegraphics[width=3cm]{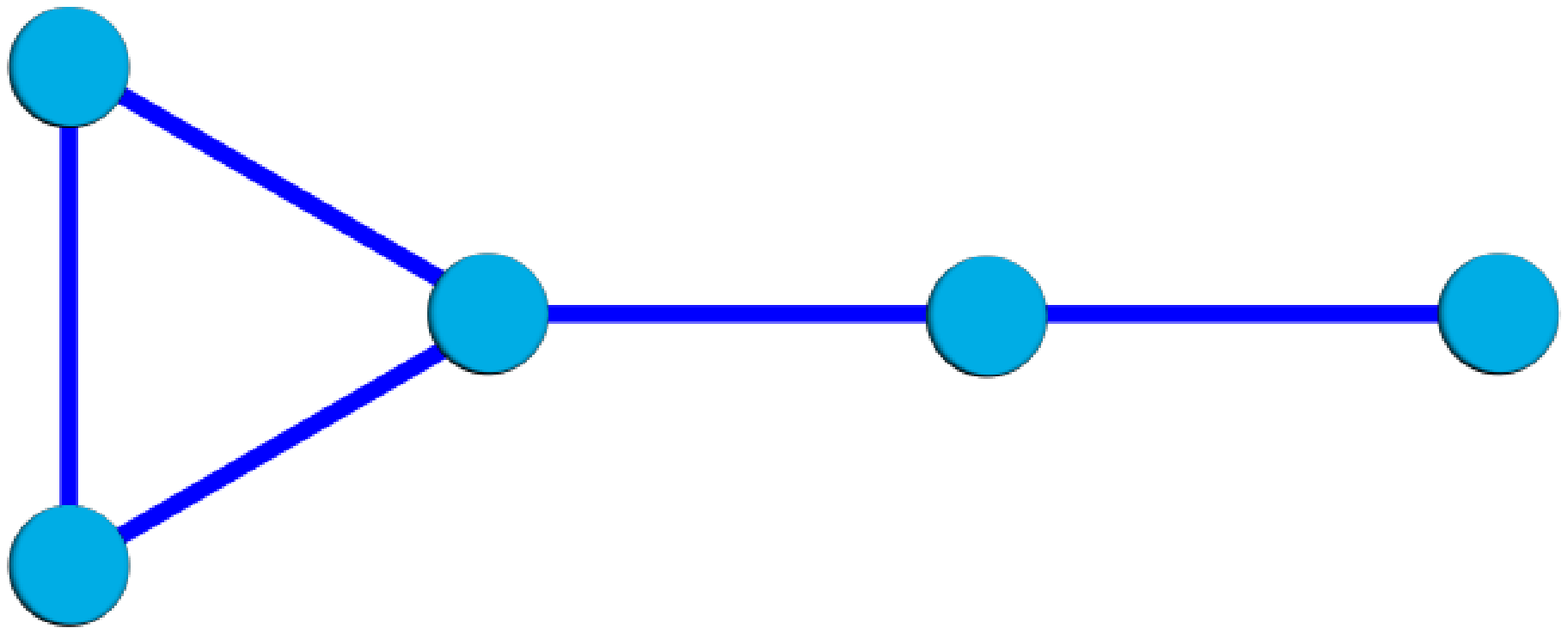}} %
\end{minipage} & %
\begin{minipage}[c]{4cm}%
\center{\includegraphics[width=2cm]{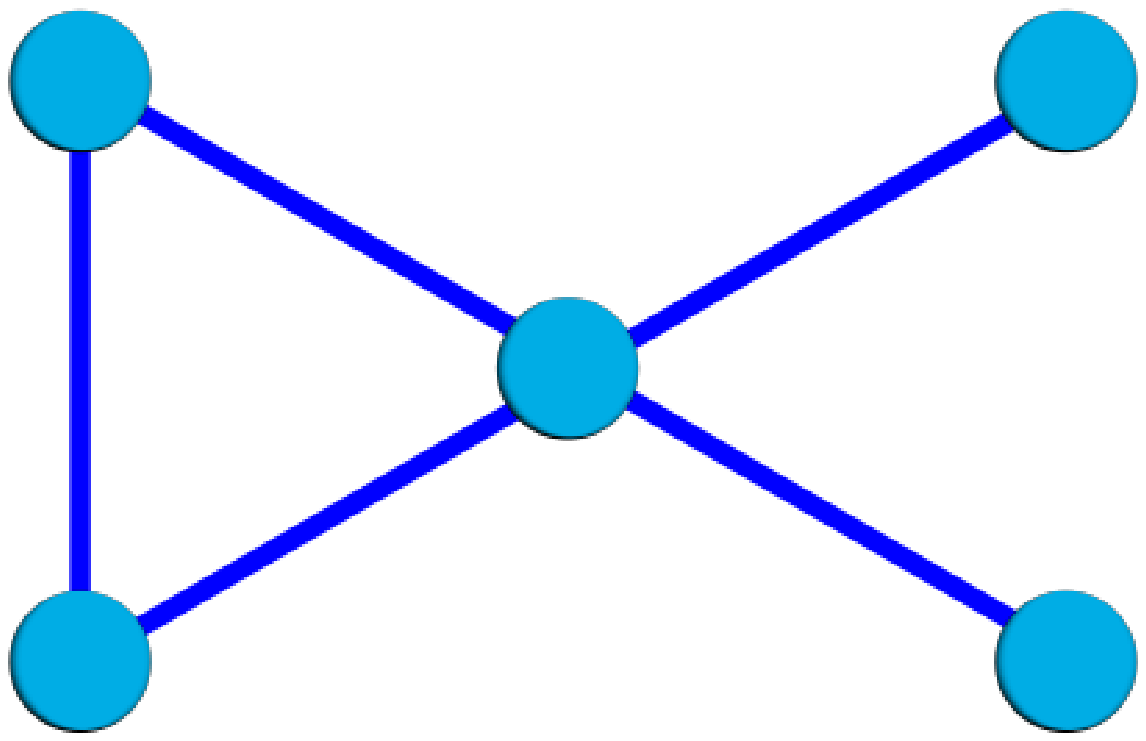}} %
\end{minipage}\tabularnewline
$S_{T1S}$  & $S_{T1D}$  & $S_{T2S}$\tabularnewline
 &  & \tabularnewline
\begin{minipage}[c]{4cm}%
\center{\includegraphics[width=1.7cm]{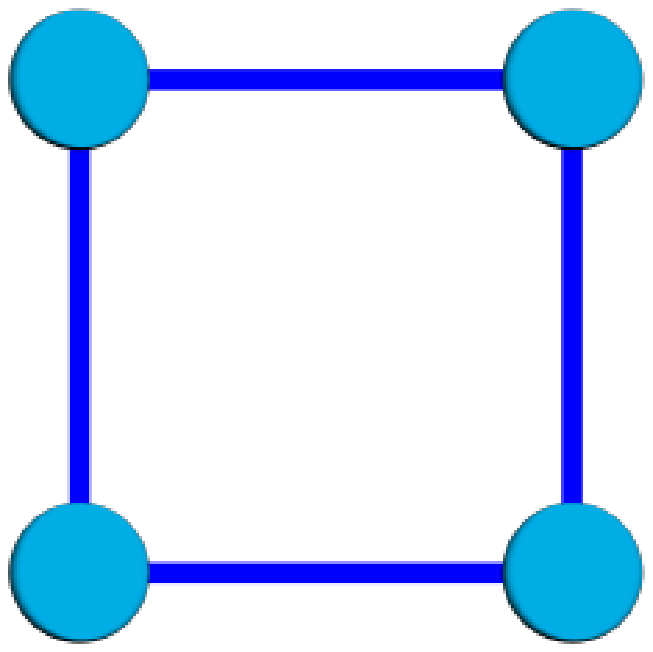}} %
\end{minipage} & %
\begin{minipage}[c]{4cm}%
\center{\includegraphics[width=1.7cm]{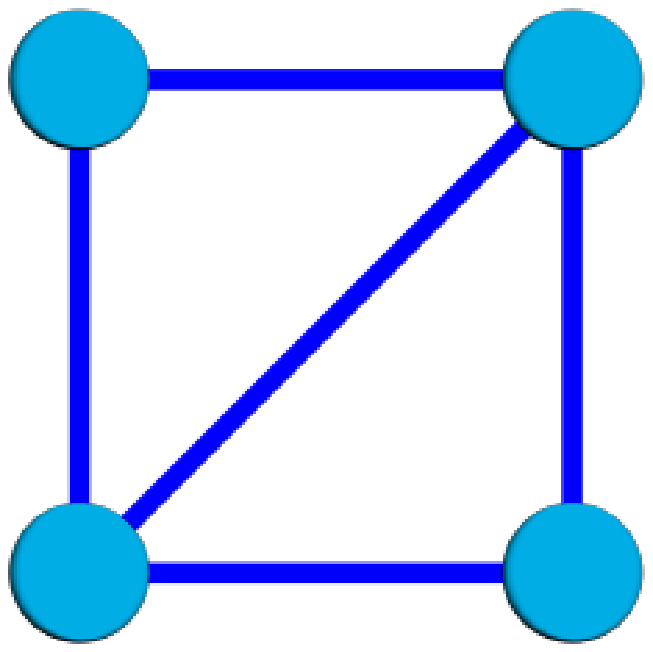}} %
\end{minipage} & %
\begin{minipage}[c]{4cm}%
\center{\includegraphics[width=2.5cm]{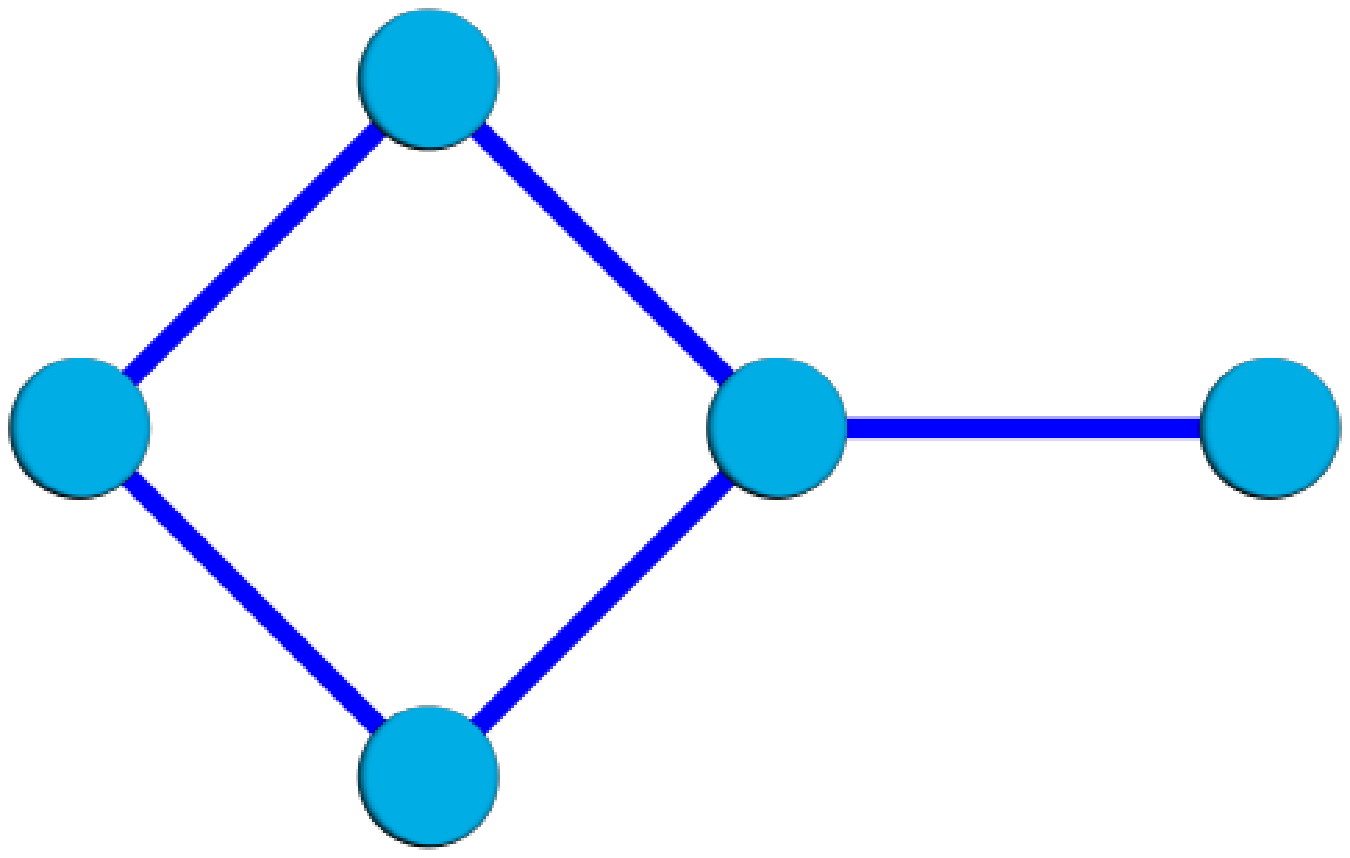}} %
\end{minipage}\tabularnewline
$C_{4}$  & $S_{C/}$  & $S_{C1S}$\tabularnewline
 &  & \tabularnewline
\begin{minipage}[c]{4cm}%
\center{\includegraphics[width=2cm]{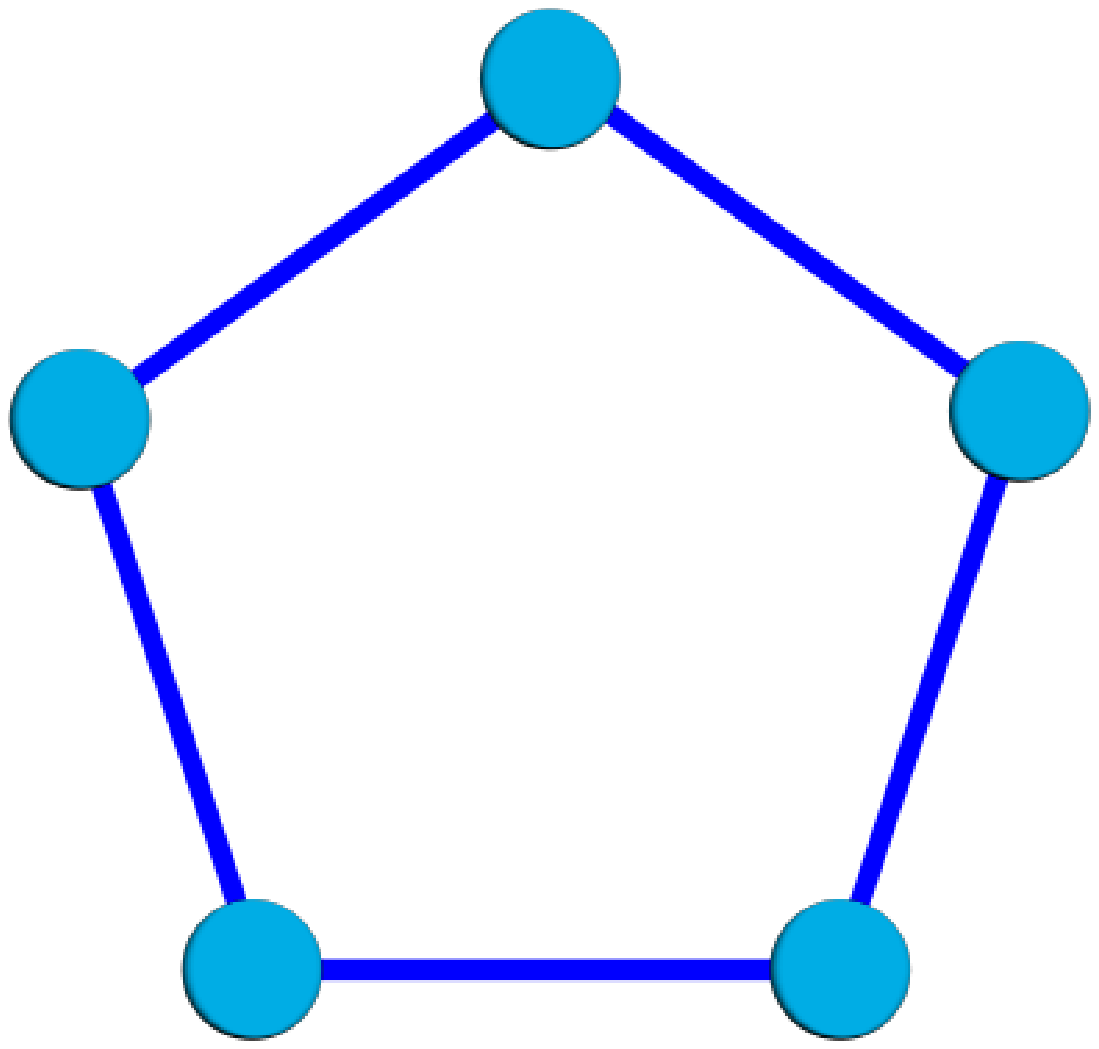}} %
\end{minipage} & %
\begin{minipage}[c]{4cm}%
\center{\includegraphics[width=2cm]{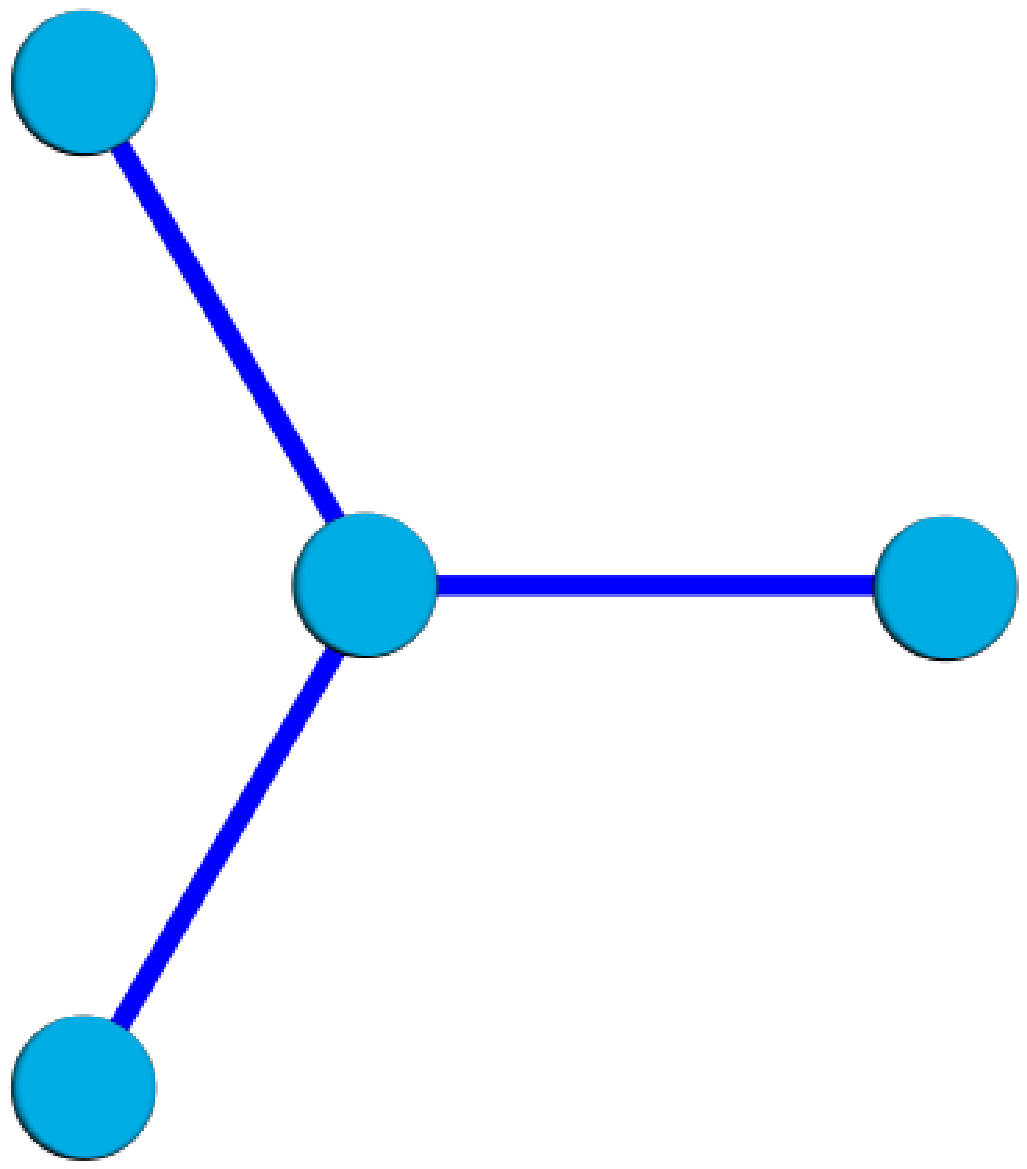}} %
\end{minipage} & %
\begin{minipage}[c]{4cm}%
\center{\includegraphics[width=2.5cm]{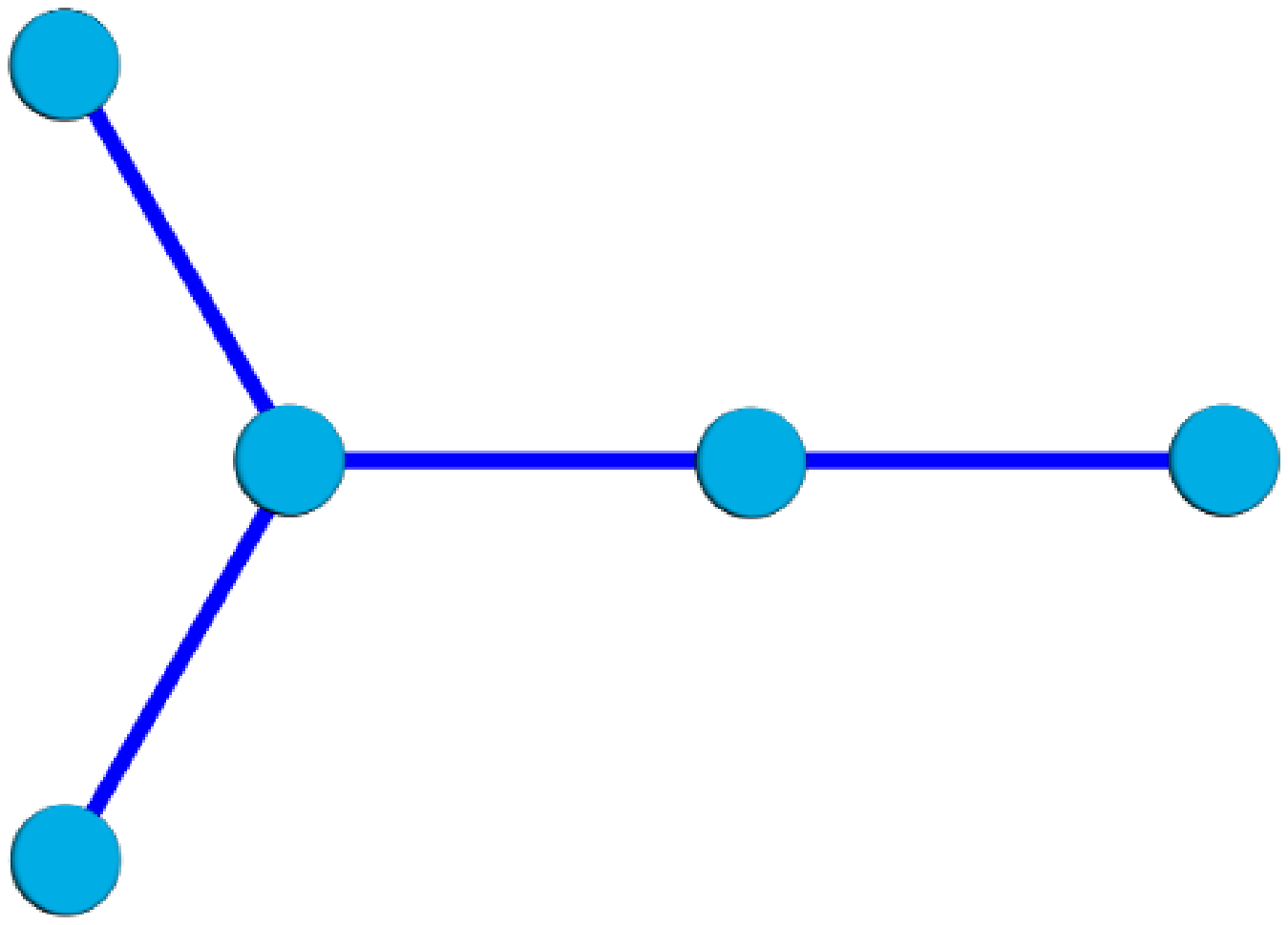}} %
\end{minipage}\tabularnewline
$C_{5}$  & $S_{1,3}$  & $S_{Y}$\tabularnewline
 &  & \tabularnewline
\end{tabular}\caption{Collection of subgraphs in Eq.~\ref{eq:R}, excluding the paths $P_{2}$, $P_{3}$, $P_{4}$, $P_{5}$, and the cycle $C_{3}$.}
\label{tab:subgraphs}
\end{table}


\section{Main Result}

\label{Sec:Main}

Our main result here consists on the determination of the two-walks degree
assortativity of a graph in terms of contributing subgraphs of the
graph. This allows us to understand this quantity in structural terms
for the analysis of real world systems in further sections of this
work. 
\begin{thm}
Let $G$ be a simple graph. Then, $G$ is, in terms of two-walks degree,

i) assortative if $R>2|P_{2}||P_{2/1}|\left(|P_{3/2}|+C+2\right)^{2}$,

ii) neutral if $R=2|P_{2}||P_{2/1}|\left(|P_{3/2}|+C+2\right)^{2}$,

iii) disassortative if $R<2|P_{2}||P_{2/1}|\left(|P_{3/2}|+C+2\right)^{2}$,

where

\begin{eqnarray}
\label{eq:R}
R  =4|P_{2}|+8|P_{3}|+4|P_{4}|+2|P_{5}|+42|C_{3}|+24|C_{4}|+10|C_{5}|\\
+12|S_{1,3}|+4|S_{Y}|+22|S_{T1S}|+4|S_{T2S}|+4|S_{T1D}|+12|S_{C/}|+4\left|S_{C1S}\right| \nonumber,
\end{eqnarray}
and $|P_{r/s}|=|P_{r}|/|P_{s}|$ and $C=3|C_{3}|/|P_{2}|$. 
\end{thm}
First, we prove that the denominator of the expression (\ref{eq:r_R_matrix})
is always non-negative. 

\begin{lem}
\label{lem:positive_denom} Let $G$ be a connected simple graph with
adjacency matrix $A$. Let $\vec{k}$ and $\vec{k}_{2}$ be vectors
of the nodes degrees and a vector of nodes two-walks degrees, respectively.
Then,

\begin{eqnarray}
\left(\vec{k}_{2}\circ\vec{k}_{2}\right)^{T}A\vec{1}-\left(\vec{1}^{T}A\vec{1}\right)^{-1}\left(\vec{1}^{T}A^{3}\vec{1}\right)^{2}\geq0
\end{eqnarray}

\noindent where $m$ is the network's number of edges, $\vec{1}$
is an all-ones vector and $\circ$ denotes the Hadamard product. 
\end{lem}
\begin{proof}
By the Cauchy-Bunyakovsky-Schwarz inequality:

\begin{align}
\left(\vec{1}^{T}A\vec{1}\right)^{-1}\left(\vec{1}^{T}A^{3}\vec{1}\right)^{2} & =\left(\vec{1}^{T}A\vec{1}\right)^{-1}\left(\sum_{i=1}^{V}k_{i}\tilde{k}_{i}\right)^{2}
\leq\left(\vec{1}^{T}A\vec{1}\right)^{-1}\left(\sum_{i=1}^{V}k_{i}^{2}\right)
\left(\sum_{i=1}^{V}\tilde{k}_{i}^{2}\right)\\
 & =\frac{2|P_{1}|+2|P_{2}|}{2|P_{1}|}\sum_{i=1}^{V}\tilde{k}_{i}^{2}=\sum_{i=1}^{V}\tilde{k}_{i}^{2}\left(\frac{1}{n}\frac{\sum_{i=1}^{V}k_{i}^{2}}{\sum_{i=1}^{V}k_{i}}\right). \nonumber
\end{align}
Then, we have

\begin{align}
\left(\vec{k}_{2}\circ\vec{k}_{2}\right)^{T}A\vec{1}-\left(\vec{1}^{T}A\vec{1}\right)^{-1}\left(\vec{1}^{T}A^{3}\vec{1}\right)^{2} & \geq\sum_{i=1}^{V}\tilde{k}_{i}^{2}\left(k_{i}-\frac{1}{n}\frac{\sum_{j=1}^{V}k_{j}^{2}}{\sum_{j=1}^{V}k_{j}}\right)\\
 & =\sum_{i=1}^{V}\tilde{k}_{i}^{2}\frac{nk_{i}\sum_{j=1}^{V}k_{j}\left(nk_{i}-k_{j}\right)}{n\sum_{j=1}^{V}k_{j}}.\nonumber
\end{align}

As $G$ is a connected simple graph, $k_{i}\geq1$ and the maximum
degree in the graph is $n-1$, then, $nk_{i}-k_{j}\geq0$, and hence
the last term is always greater than or equal to zero, which proves
the result. 
\end{proof}

What remains now for the proof of the main result is to express the
numerator $N\left(r\right)$ of the Pearson coefficient of the
two-walks degree - two-walks degree correlation in terms of subgraphs of the
graph (reminding that when the denominator is equal to zero, the Pearson Correlation coefficient is not defined). We can write $N\left(r\right)$ as follow 
\begin{eqnarray}
N\left(r\right)=\frac{1}{m}\sum_{\left(i,j\right)\in E}\tilde{k}_{i}\tilde{k}_{j}-\left\{ \frac{1}{m}\sum_{\left(i,j\right)\in E}\frac{1}{2}\left[\tilde{k}_{i}+\tilde{k}_{j}\right]\right\} ^{2}\label{eq:r_tilde_k}
\end{eqnarray}

\noindent where $\tilde{k}_{i}$ and $\tilde{k}_{j}$ are the
two-walks degrees of nodes $i$ and $j$, respectively, located at both
ends of link $\left(i,j\right)\in E$. We can now rewrite the sums in Eq. (\ref{eq:r_tilde_k})
as:

\begin{equation}
\sum_{\left(i,j\right)\in E}\tilde{k}_{i}\tilde{k}_{j}=\frac{1}{2}\vec{k}_{2}^{T}A\vec{k}_{2},
\end{equation}

\begin{equation}
\sum_{\left(i,j\right)\in E}\left(\tilde{k}_{i}+\tilde{k}_{j}\right)=\vec{k}_{2}^{T}\vec{k}.
\end{equation}


Let us now find the expressions for the two terms contributing to
$N\left(r\right)$. The first is given by

\begin{eqnarray}
& \frac{1}{m}\sum_{\left(i,j\right)\in E}\tilde{k}_{i}\tilde{k}_{j} =& \frac{1}{2m}\left(2|P_{1}|+12|P_{2}| +12|P_{3}|+4|P_{4}|+2|P_{5}|\right. +\\
& & +54|C_{3}|+24|C_{4}|+10|C_{5}|+12|S_{1,3}|+22|S_{T1S}| + \nonumber\\
& & \left.+4|S_{Y}|+4|S_{T2S}|+4|S_{T1D}|+12|S_{C/}|+4|S_{C1S}|\right), \nonumber
\end{eqnarray}
where $|P_{4}|$ and $|P_{5}|$ are the number of paths of order $4$
and $5$, respectively, and $|S_{Y}|$ is the number of fragments
$S_{Y}$ which are illustrated in Table \ref{tab:subgraphs}. We will
give formulas for calculating these fragments for the sake of completeness
of the paper.

For the second term contributing to $N\left(r\right)$ we have

\begin{equation}
\left\{ \frac{1}{m}\sum_{\left(i,j\right)\in E}\frac{1}{2}\left[\tilde{k}_{i}+\tilde{k}_{j}\right]\right\} ^{2}=\frac{1}{\left(2m\right)^{2}}\left(2|P_{1}|+4|P_{2}|+2|P_{3}|+6|C_{3}|{}^{2}\right).
\end{equation}

Thus, we can rewrite $N\left(r\right)$ as:

\begin{equation}
N\left(r\right)=\frac{1}{2m}\left(R-2|P_{2}||P_{2/1}|\left(|P_{3/2}|+C+2\right)^{2}\right),\label{eq:Nr_sub}
\end{equation}
which proves the main result.

Let us now give the formulas for calculating the subgraphs remaining
in the expression of the two-walks degree assortativity which have not
been previously defined. The proofs of these results are based on
the strategy developed and explained in \cite{Estrada_Knight} and
are not given here as they are lengthly and technical.

\begin{lem}
\label{lem:S_y} Let $G=\left(V,E\right)$ be a simple graph. Let
$k_{i}$ be the degree of the vertex $i$. Let $|S_{Y}|$ be the number
of subgraphs $S_{Y}$ (see Table \ref{tab:subgraphs}). Then, 
\end{lem}
\begin{equation}
|S_{Y}|=\sum_{\left(i,j\right)\in E}\left\{ \left(\begin{array}{c}
k_{i}-1\\
2
\end{array}\right)\left(k_{j}-1\right)+\left(\begin{array}{c}
k_{j}-1\\
2
\end{array}\right)\left(k_{i}-1\right)\right\} -2|S_{T1S}|.
\end{equation}

\begin{lem}
Let $G=\left(V,E\right)$ be a simple graph. Then, the number of subgraphs
$|P_{4}|$ and \textup{$|P_{5}|$} in $G$ are given by, respectively,

\begin{equation}
|P_{4}|=\frac{1}{2}\vec{1}^{T}A^{4}\vec{1}-|P_{1}|-4|P_{2}|-2|P_{3}|-9|C_{3}|-4|C_{4}|-6|S_{1,3}|-4|S_{T1S}|,
\end{equation}

\begin{eqnarray}
&|P_{5}| & =\frac{1}{2}\vec{1}^{T}A^{5}\vec{1}-|P_{1}|-6|P_{2}|-6|P_{3}|-2|P_{4}|-27|C_{3}|+ \\
& & -12|C_{4}|-5|C_{5}|-6|S_{1,3}|-2|S_{Y}|-11|S_{T1S}|-2|S_{T2S}|+ \nonumber\\
& & -2|S_{T1D}|-6|S_{C/}|-2|S_{C1S}|.\nonumber
\end{eqnarray}
\end{lem}


\section{Computational results}
\label{Sec:Num}
\subsection{Small graphs}

In this Section we describe the results obtained for all the 261,000
connected graphs with 9 unlabelled nodes. We have calculated the degree
and two-walks degree assortativities for these graphs (see Fig. \ref{9nodesfig}).
As we can see there is no trivial correlation between the two indices,
which indicates that the new index does not duplicate the structural
information contained in the degree assortativity and consequently
gives some new structural insights about graphs. This conclusion is
also easily obtained by considering the subgraph contributions to
both measures.

\begin{figure}
\begin{centering}
\includegraphics[width=0.6\textwidth]{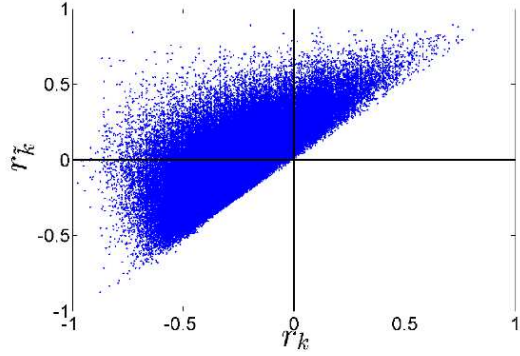} \label{9nodesfig} 
\par\end{centering}
\caption{Degree and two-walks degree assortativities for all the connected graphs
with 9 unlabelled nodes.}
\end{figure}


According to computer calculations 7\% of the networks are assortative-assortative
by both measures (AA), 60\% are disassorartive-disassortative (DD)
and 33\% are disassortative by degree and assortative by two-walks degree
(DA). The main observation is that there are no graphs which are degree assortative and two-walks degree disassortative (AD). We conjecture that these graphs
cannot exist. Computer calculations show that $R\gtrsim2|P_{2}|(|P_{3,2}|+C)(|P_{3,2}|+C+2)^{2}$.
Therefore, we can express the numerator of the neighbourhood assortativity
Eq. (\ref{eq:Nr_sub}) as follows:

\begin{eqnarray}
N\left(r\right)\gtrsim\frac{1}{m}|P_{2}|(|P_{3,2}|+C-|P_{2/1}|)(|P_{3,2}|+C+2)^{2}
\end{eqnarray}

Using the results from \cite{Estrada11}, if $r_{k}\geq0$, then
$N\left(r\right)\geq0$. The intuition behind this result is very
simple. Nodes that belong to a degree assortativitive network tend
to be linked to other nodes with similar degree. Therefore, their
two-walks degrees tend to be similar too.

Generally, the second-neighbour degree assortativity depends on the balance between
four structural factors: the weighted sum of subgraphs given by $R/|P_{2}|$,
transitivity ($C$), intermodular connectivity ($|P_{3,2}|$), relative
branching ($|P_{2,1}|$). The first three produce a positive contribution
to the two-walks degree assortativity of a network, while branching is
more likely associated with disassortative networks.

\subsection{Real-world networks}

\textcolor{black}{In this subsection we study of group of 49 real-world
networks representing systems in ecological (E), biological (B), social
(S), technological (T) and socio-economic (SE) envirnments. The networks
are described in the Appendix of this paper. We have calculated the
degree and two-walks degree assortativities for these networks (see Fig.
\ref{Scatterplot}). According to these results 14\% of the networks
are assortative-assortative (AA) according to both measures, 24\%
are disassorartive-disassortative (DD) and the majority of networks
analyzed (61\%) are diassortative-assortative (DA). This confirms
our previous observation that there are no graphs/networks which are
assortative-disassorartive (AD). The analysis of the networks according
to the functions shows the following trends: 53\% of the ecological
networks analyzed are DD, 27\% are DA and 20\% are AA; 50\% of the
social networks analyzed are DA, 30\% are AA and 20\% are DD; 80\%
of technological networks are DA, 10\% are AA and 10\% are DD. Finally,
100\% of biological networks considered are DA. They included
9 protein-protein interaction networks (PINs), 3 transcription networks
and 3 brain networks. This is a remarkable observation because it
is the only single functional class of networks which is formed by
one structural class, i.e., DA. }

\begin{figure}
\begin{centering}
\includegraphics[width=0.7\textwidth]{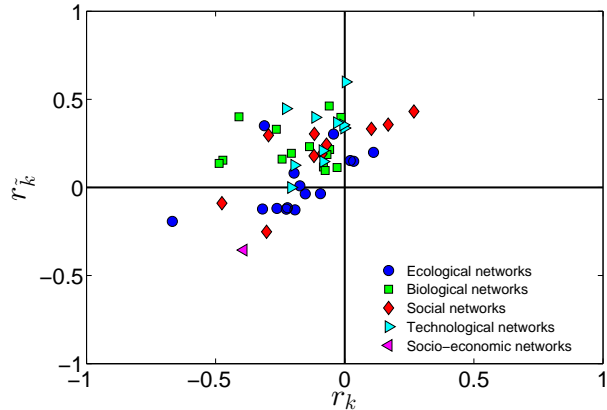} 
\par\end{centering}
\caption{Degree and two-walks degree assortativities for all the real-world networks studied in this work.}

\label{Scatterplot} 
\end{figure}

An important characteristic of our current approach is that we can
understand the structural causes for the different kinds of assortativity
in networks using the interpretation of these quantities in terms
of subgraphs of the graph. As we have seen before an important structural
feature of graphs allowing the transition from degree disassortative
to two-walks degree assortative is the presence of separators. It has
to be stressed that this is not a unique structural feature of this
kind of networks and more studies are needed to completely understand
the structural chracterization of this kind of networks. However,
it is easy to visualize the connectors in the small PIN of the bacterium
\textit{\textcolor{black}{B. subtilis}} (see left panel in Fig. \ref{Illustration networks}).
In Fig. \ref{Illustration networks} we also illustrate the degree
and two-walks degree of the nodes in the food web of ScotchBroom
and in the transcription network of \textit{\textcolor{black}{E. coli}}.
All of them displaying DA characteristics.

\begin{figure}
\begin{centering}
\includegraphics[width=0.3\textwidth]{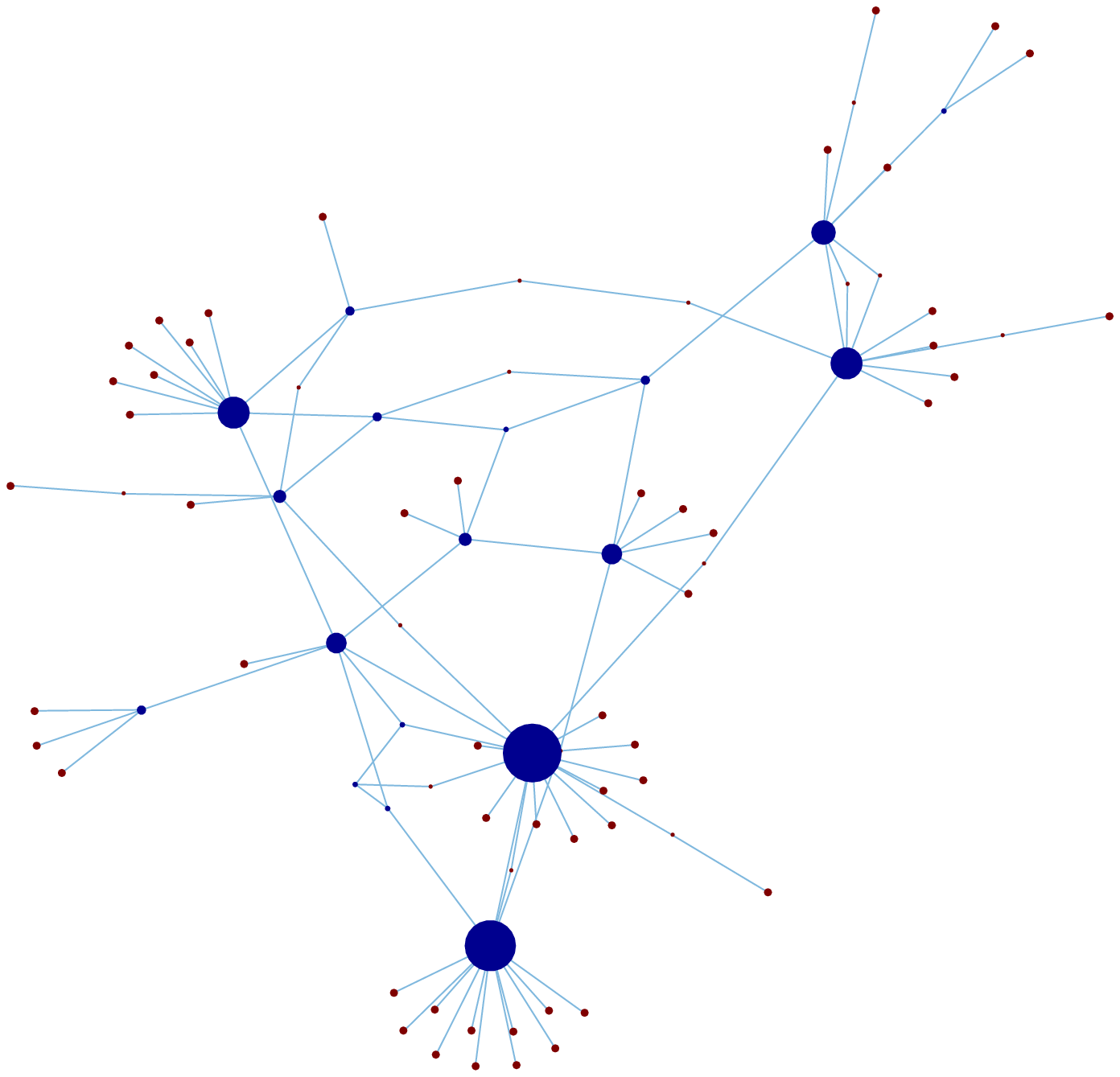}\includegraphics[width=0.3\textwidth]{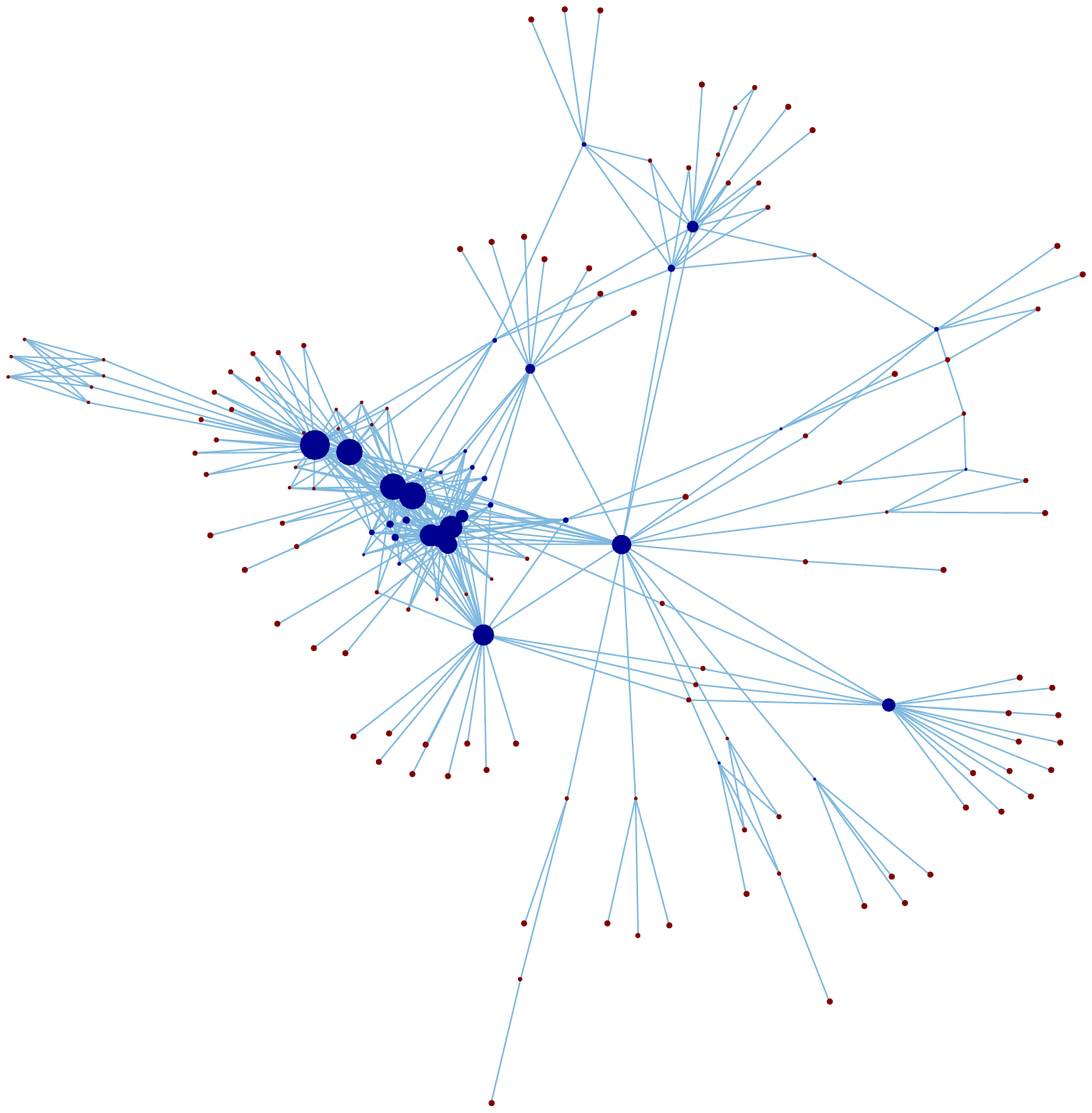}
\includegraphics[width=0.3\textwidth]{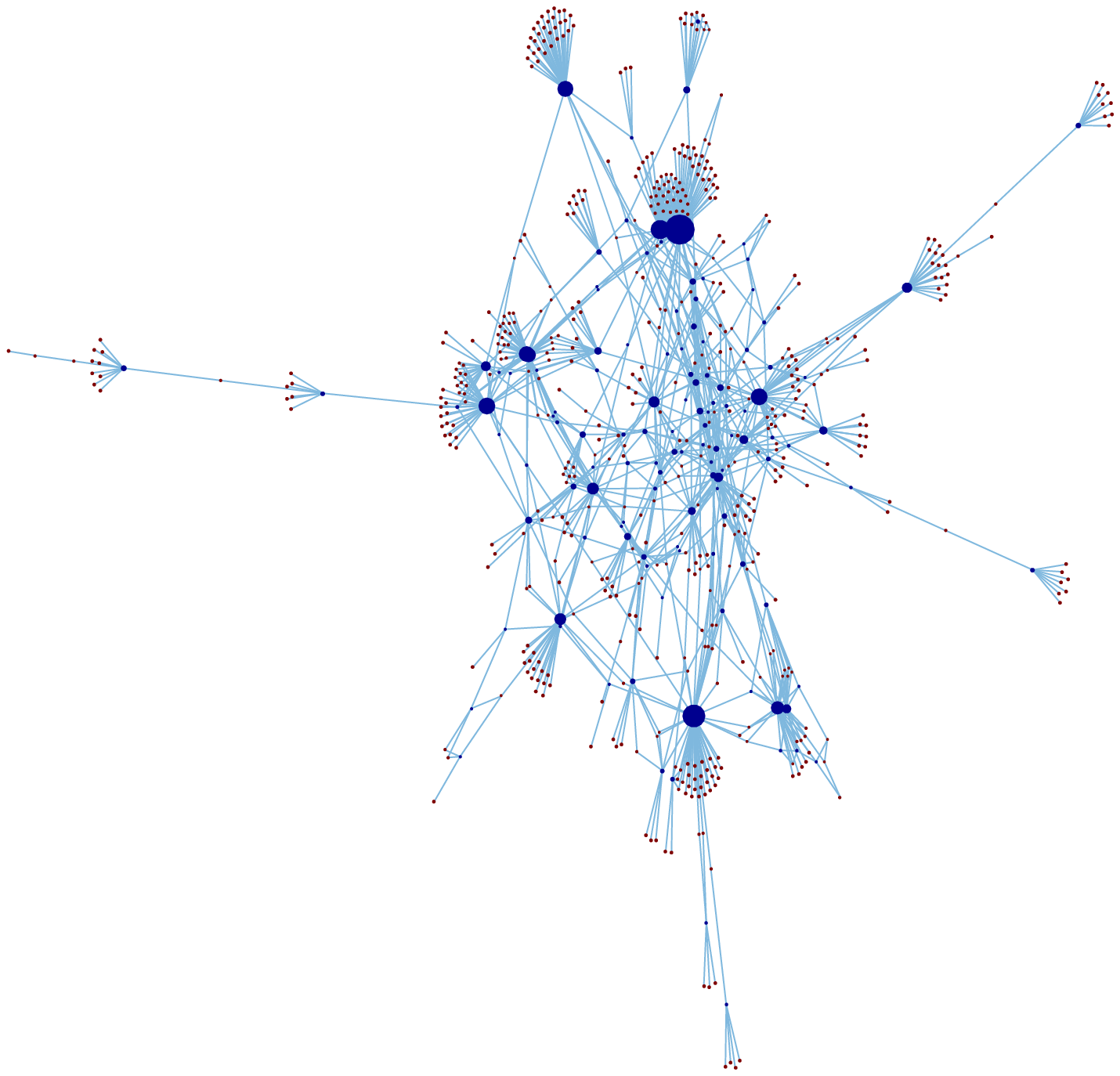} 
\par\end{centering}
\begin{centering}
\includegraphics[width=0.3\textwidth]{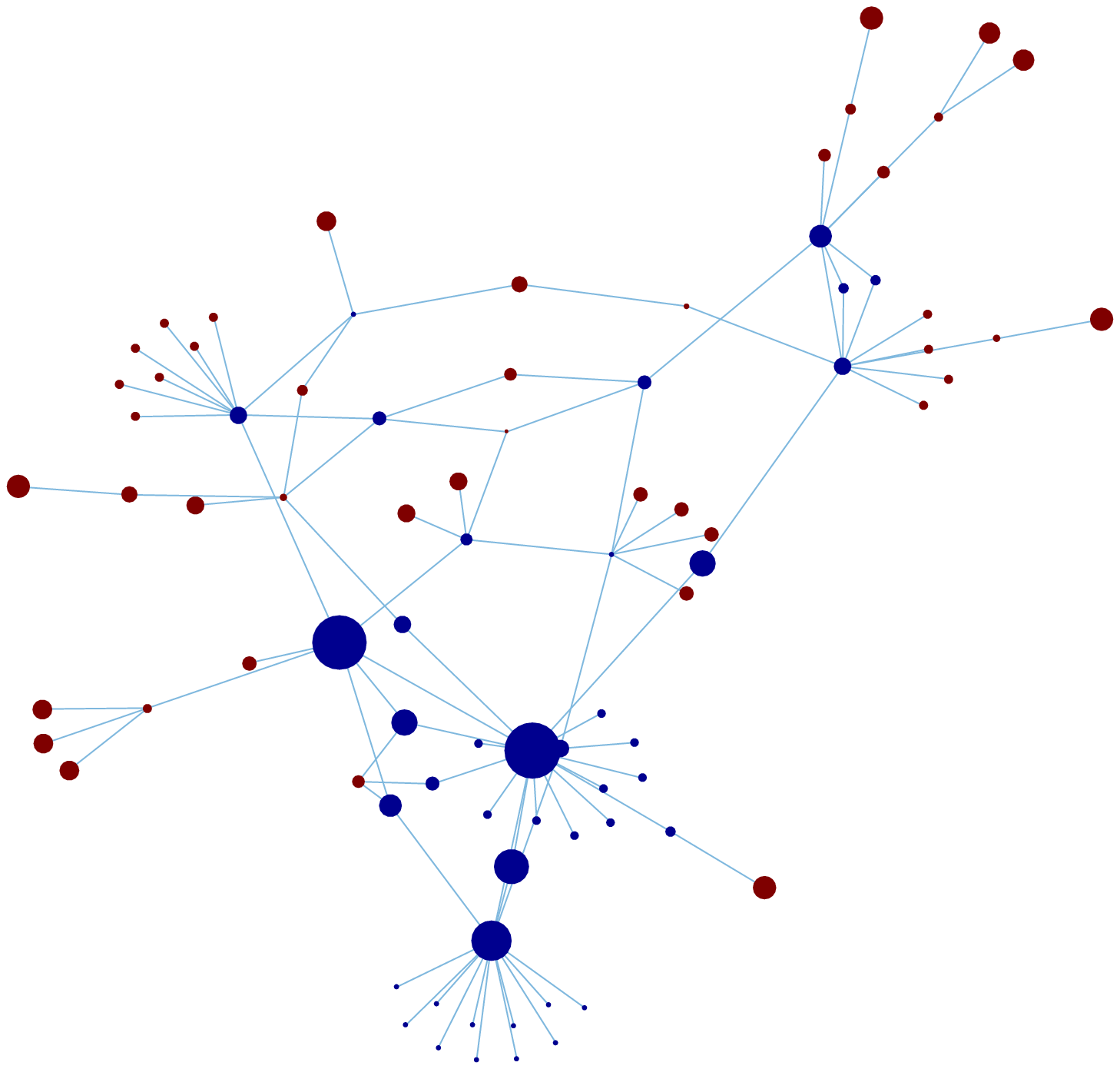}\includegraphics[width=0.3\textwidth]{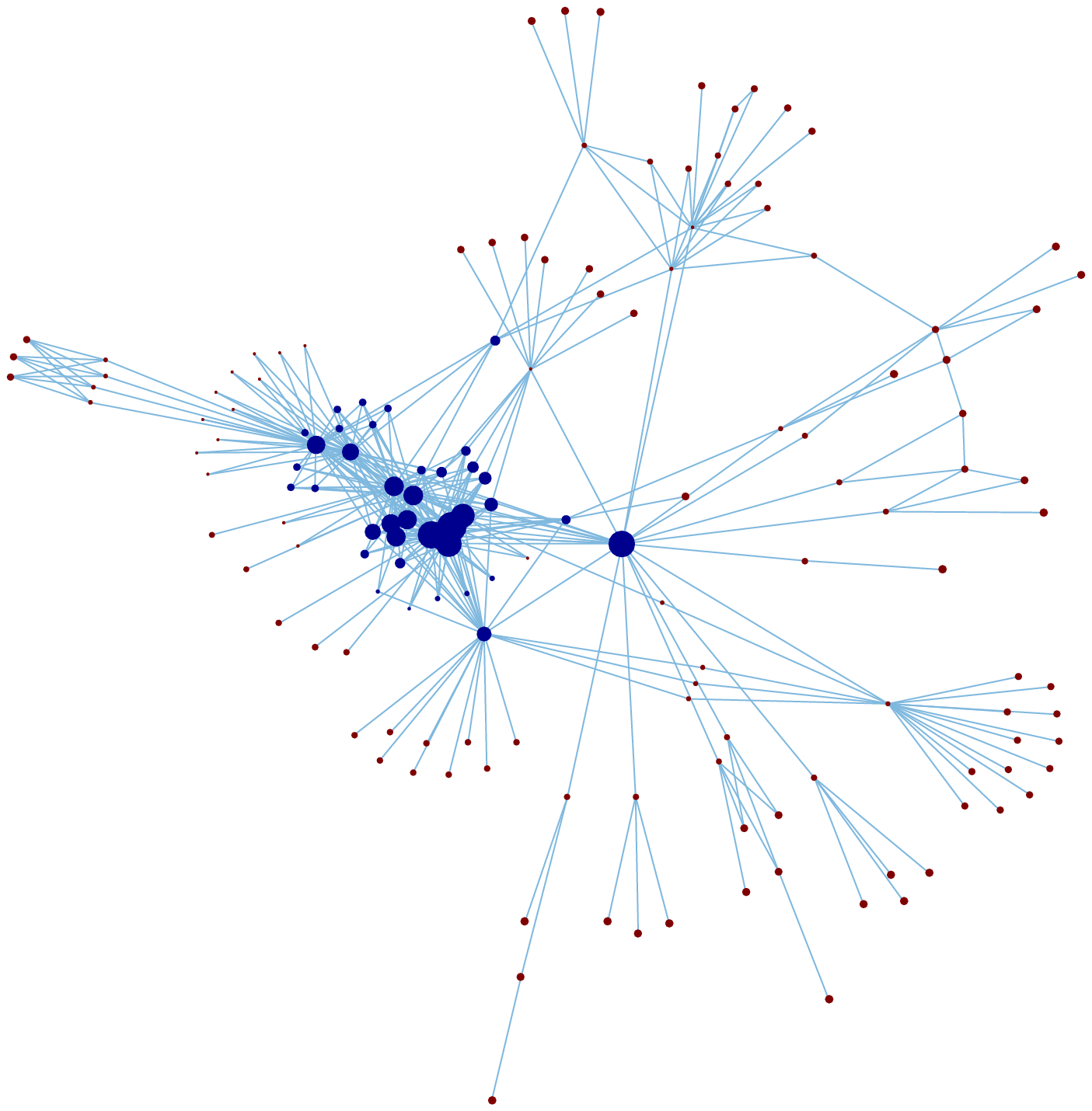}\includegraphics[width=0.3\textwidth]{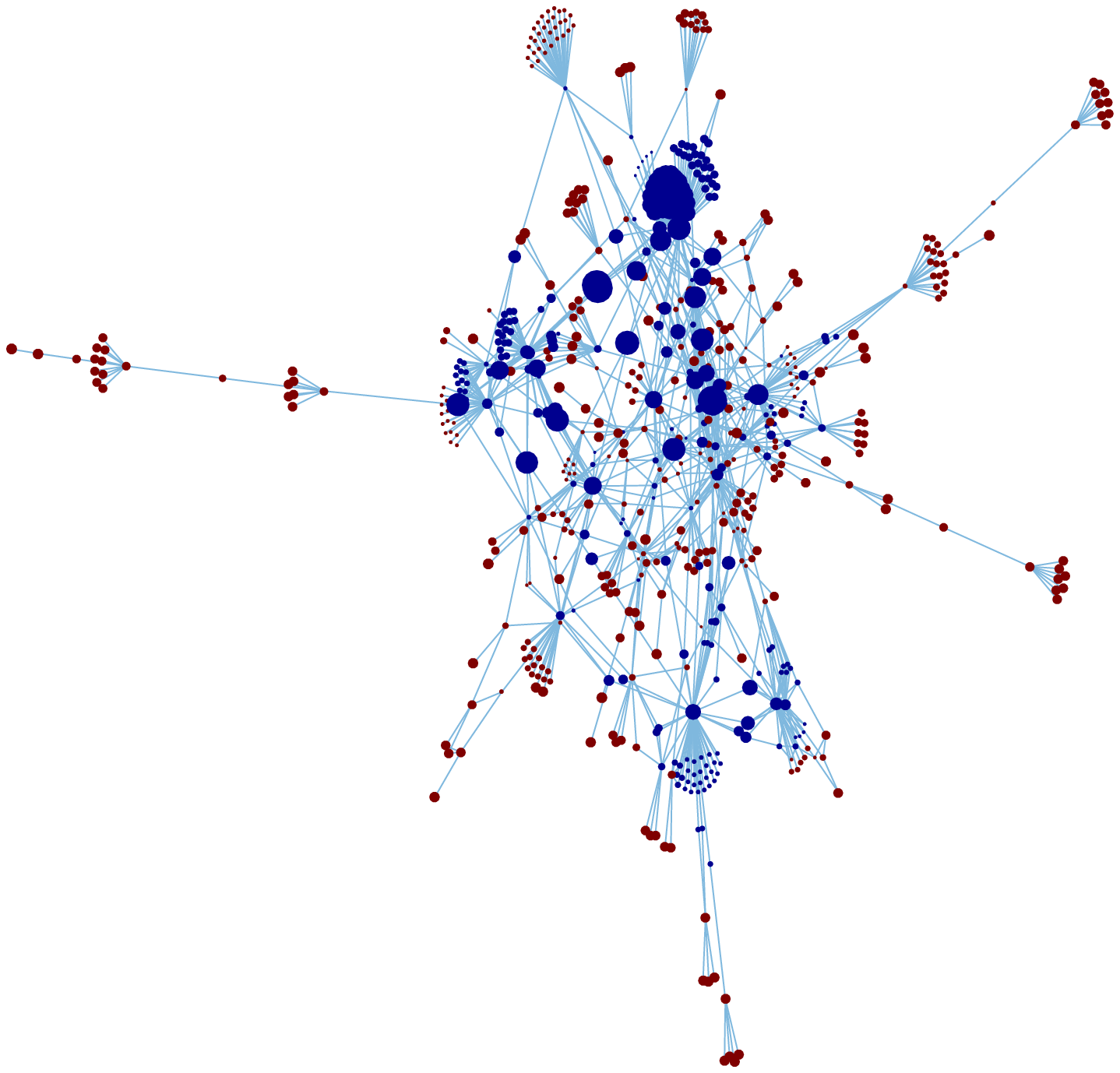} 
\par\end{centering}
\caption{Illustration of the differences between the degrees and mean degree of every node (top panels) and the same for the two-walks degrees (bottom panels) in the protein interaction network \textit{B. subtilis}, food web of Scotch Broom, and transcription network of yeast from left to right.}
\label{Illustration networks} 
\end{figure}

\section{Conclusions}
\label{Sec:Conclusion}
Here we have proposed an extension of the concept of degree assortativity
to one that account for thecorrelation between the degrees
of the nodes and their nearest neighbours in graphs and networks. This measure, here named the
two-walks degree assortativity, is expressed in terms of subgraphs of
the graph. As we have proved here there are a few more fragments contributing
to the two-walks degree assortativity than to the degree assortativity.
This clearly indicates that the new quantity accounts for more structural
information than the previous one. We have seen that both measures
are not linearly correlated neither for all the connected graphs
with 9 nodes nor for real-world networks. Further studies are needed
to understand the role of this quantity in the study of real-world
problems, as we have seen here, there are some apparently universal
features of some classes of networks in relation to this quantity.
For instance, all real-world biological networks studied here are
degree disassortative but two-walks assortative. The implications
of this observation for the study of the biological processes taking
place on these networks is far beyond the scope of this work.

\section*{Appendix: Real-world network dataset}

The real-world networks used in this paper belong to different domains:
ecological (includes food webs and ecosystems), social (networks of
friendships, communication networks, corporate relationships), technological
(internet, transport, software development networks), informational
(vocabulary networks, citations) and biological (protein-protein interaction
networks, transcriptional regulation networks). The dataset comprises
networks of different sizes, ranging from $N=29$ to $N=4941$ nodes.
The networks are listed in Table~\ref{tab:tab1}.




\begin{table}[ht]
\caption{Dataset of real-world networks: network name, domain, $N$ number
of nodes, $m$ number of links, reference, degree and two-walks degree assortative coefficients.}
\begin{tabular}{c l c c c c c c}
\hline 
No.  & Dataset  & Domain  & N  & m  & Ref. & $r_{k}$ &  $r_{\tilde{k}}$  \tabularnewline
\hline
1  & Drosophila PIN  & biological  & 3039  & 3715  & \cite{giot2003protein}  & -0.060 &	0.462 \tabularnewline 
2  & Hpyroli  & biological  & 710  & 1396  & \cite{refn16} & -0.243 &	0.161 \tabularnewline 
3  & KSHV  & biological  & 50  & 122  & \cite{refn5} & -0.058 &	0.215 \tabularnewline 
4  & MacaqueVisualCortex  & biological  & 30  & 190  & \cite{sporns2004motifs} & -0.030 & 	0.113 \tabularnewline 
5  & Malaria PIN  & biological  & 229  & 604  & \cite{refn32} & -0.083 &	0.116 \tabularnewline 
6  & Neurons  & biological  & 280  & 1973  & \cite{refn14} & -0.069 &	0.187 \tabularnewline 
7  & PIN-Afulgidus  & biological  & 32  & 38  & \cite{refn54} & -0.472 &	0.154 \tabularnewline 
8  & Pin-Bsubtilis  & biological  & 84  & 98  & \cite{refn53} & -0.486 &	0.136 \tabularnewline 
9  & PIN-Ecoli  & biological  & 230  & 695  & \cite{refn27} & -0.015&	0.397 \tabularnewline 
10  & PIN-Human  & biological  & 2783  & 6438  & \cite{rual2005towards} & -0.137 &	0.231 \tabularnewline 
11  & Trans-Ecoli  & biological  & 328  & 456  & \cite{refn4} & -0.265 &	0.330\tabularnewline 
12  & Transc-yeast  & biological  & 662  & 1062  & \cite{refn4} & -0.410 &	0.401 \tabularnewline 
13  & Trans-urchin  & biological  & 45  & 80  & \cite{refn4} & -0.207 &	0.194 \tabularnewline 
\hline
14  & Benguela  & ecological  & 29  & 191  & \cite{refn21} & 0.0211 &	0.153 \tabularnewline 
15  & BridgeBrook  & ecological  & 75  & 547  & \cite{refn35} & -0.668 &	-0.193 \tabularnewline 
16  & Canton  & ecological  & 108  & 708  & \cite{refn33} & -0.226 &	-0.123 \tabularnewline 
17  & Chesapeake  & ecological  & 33  & 72  & \cite{refn20} & -0.196 &	0.081 \tabularnewline 
18  & Coachella  & ecological  & 30  & 261  & \cite{refn49} & 0.0347 &	0.148 \tabularnewline 
19  & ElVerde  & ecological  & 156  & 1441  & \cite{refn38} &-0.174 &	0.009 \tabularnewline 
20  & ReefSmall  & ecological  & 50  & 524  & \cite{refn1} & -0.193 & -0.127 \tabularnewline 
21  & ScotchBroom  & ecological  & 154  & 370  & \cite{refn26} & -0.311 &	0.350 \tabularnewline 
22  & Shelf  & ecological  & 81  & 1476  & \cite{link2002doe}  & -0.094 &	-0.035 \tabularnewline 
23  & Skipwith  & ecological  & 35  & 364  & \cite{refn29} & -0.319	& -0.122 \tabularnewline 
24  & StMarks  & ecological  & 48  & 221  & \cite{refn30} & 0.111 &	0.199 \tabularnewline 
25  & StMartin  & ecological  & 44  & 218  & \cite{refn2} & -0.153 & -0.0365 \tabularnewline 
26  & Stony  & ecological  & 112  & 832  & \cite{refn36} & -0.222 &	-0.115 \tabularnewline 
27  & Ythan1  & ecological  & 134  & 597  & \cite{refn15} & -0.263 &	-0.119 \tabularnewline 
\hline
28  & World Trade  & economic & 80  & 875  & \cite{batagelj2006analysis} & -0.392 &	-0.355 \tabularnewline 
\newline
29  & SmallW  & informational  & 233  & 994  & \cite{refn9} & -0.303 &	-0.251 \tabularnewline 
\hline
30  & ColoSPG  & social  & 324  & 347  & \cite{refn51} & -0.295 &	0.296\tabularnewline 
31  & CorporatePeople  & social  & 1586  & 13126  & \cite{refn43} & 0.268 & 0.431 \tabularnewline 
32  & Dolphins  & social  & 62  & 159  & \cite{refn37} & -0.044 &	0.303 \tabularnewline 
33  & Drugs  & social  & 616  & 2012  & \cite{batagelj2006analysis} & -0.117 & 	0.304 \tabularnewline 
34  & Hi-tech  & social  & 33  & 91  & \cite{refn23} & -0.087 &	0.191 \tabularnewline 
35  & Geom  & social  & 3621  & 9461  & \cite{batagelj2006analysis}  & 0.168 &	0.356\tabularnewline 
36  & PRISON-Sym  & social  & 67  & 142  & \cite{refn25} & 0.103 &	0.332 \tabularnewline 
37  & Sawmill  & social  & 36  & 62  & \cite{refn10} & -0.071 &	0.243 \tabularnewline 
38  & social3  & social  & 32  & 80  & \cite{refn31} & -0.119 &	0.179 \tabularnewline 
39  & Zackar  & social  & 34  & 78  & \cite{refn24} & -0.476 &	-0.089 \tabularnewline 
\hline
40  & electronic1  & technological  & 122  & 189  & \cite{refn40} & -0.002 &	0.337 \tabularnewline 
41  & electronic2  & technological  & 252  & 399  & \cite{refn40} & -0.006 &	0.355 \tabularnewline 
42  & electronic3  & technological  & 512  & 819  & \cite{refn40} & -0.030 &	0.367 \tabularnewline 
43  & Power grid  & technological  & 4941  & 6594  & \cite{watts1998collective}  & 0.003 &	0.599 \tabularnewline 
44  & Software Abi  & technological  & 1035  & 1736  & \cite{refn7} & -0.086 &	0.208 \tabularnewline 
45  & Software Digital  & technological  & 150  & 198  & \cite{refn7} & -0.228 &	0.447 \tabularnewline 
46  & Software Mysql  & technological  & 1480  & 4221  & \cite{refn7} & -0.083 &	0.147 \tabularnewline 
47  & Software-XMMS  & technological  & 971  & 1809  & \cite{refn7} & -0.114 &	0.397 \tabularnewline 
48  & Software-VTK  & technological  & 771  & 1369  & \cite{refn7} & -0.195 &	0.126 \tabularnewline 
49  & USA Air 97  & technological  & 332  & 2126  & \cite{refn9} & -0.208 &	-0.000 \tabularnewline 
\hline
\end{tabular}
\label{tab:tab1}
\end{table}

\end{document}